\documentclass[reqno]{amsart}
\usepackage{bbm}
\usepackage{amsfonts}
\usepackage{mathrsfs}
\usepackage{hyperref}
\usepackage{amssymb}
\usepackage{CJK}
 \newtheorem{thm}{Theorem}[section]
 \newtheorem{cor}[thm]{Corollary}
 \newtheorem{prop}[thm]{Corollary}
 
 \newtheorem{conj}[thm]{Conjecture}
 \newtheorem{lem}[thm]{Lemma}

 \newtheorem{defn}[thm]{Definition}
 \theoremstyle{remark}
 
 \numberwithin{equation}{section}
\DeclareMathOperator{\sign}{sign}

\begin{document}
\title[On the average sensitivity of the weighted sum function]
  { On the average sensitivity of the weighted sum function}

\author{Jiyou Li}
\address{Department of Mathematics, Shanghai Jiao Tong University, Shanghai, P.R. China}
\email{lijiyou@sjtu.edu.cn}

%\author{Daqing Wan}
%\address{Department of Mathematics, University of California, Irvine, CA 92697-3875, USA}
%\email{dwan@math.uci.edu}

 %\author[ J. Y. Li, D. Q. Wan]
 %{Jiyou Li, Daqing Wan}

%\date{2006 10 10}
%\thanks{Submitted September 8, 2005. Published January 19, 2006.}
\thanks{This work is supported by the National Science Foundation of China
(11001170).}
%This work is supported by NSFC, grants
%National Science Foundation of China
%  and \hfill\break\indent
%Y2005A07 from Natural Science Foundation, Shandong Province,
%China} \subjclass[2000]{34B16} \keywords{Singular sublinear
%boundary-value problem; positive solution; \hfill\break\indent
%fixed point theorem; cone; higher order differential equation}

\begin{abstract}

In this paper we obtain the bound on the average sensitivity of the
weighted sum function. This confirms a conjecture of Shparlinski. We
also compute the weights of the weighted sum functions and show that
they are almost balanced.

%As an illustration of applications, we give a combinatorial proof on
%a problem raised by Stanley.

%In particular,  we obtain some positive results in the decoding
%problems arising from Reed-Solomon codes this together with the by
%by using the Weil bound.
%We also try to apply our method to the calculations of the number of
%permutation polynomials over finite fields.
\end{abstract}

\maketitle \numberwithin{equation}{section}
\newtheorem{theorem}{Theorem}[section]
\newtheorem{lemma}[theorem]{Lemma}
\newtheorem{example}[theorem]{Example}
\allowdisplaybreaks

\section{Introduction}

A weighted sum function, also known as laced Boolean function, is
defined in terms of certain weighted sums in the residue ring modulo
a prime. Explicitly, it can be defined as follows \cite{SZ}. Let $n$
be a positive integer and $p$ is the least prime number that is no
less than $n$. For $X=(x_1,x_2,\dots,x_n)\in \mathbb{Z}_2^n$, we
define $s(X)$ by the least positive integer of $\sum_{k=1}^n kx_k$
modular $p$, i.e.,
$$s(X)\equiv \sum_{k=1}^n kx_k(\!\!\!\mod p), 1\leq s(X) \leq p.$$
 We define that
$$f(X)=\left\{
\begin{array}{ll}
  x_{s(X)}, \ \  1\leq s(X)\leq n;\\
  x_1, \ \  \hbox{otherwise.}\\
    \end{array}
    \right.
    $$

  This function was first studied by P.
Savick\'{y} and S. \v{Z}\'{a}k \cite{SZ} in their study of read-once
branching programs. It has also been used  for several more
complexity theory applications by M. Sauerhoff \cite{SD,S}. For
instance, in \cite{S} a certain modification of the same function
has been used to prove that quantum read-once branching programs are
exponentially more powerful than classical read-once branching
programs.% Shparlinski gave a nontrivial bound on the Fourier
%coefficients of
% $f(x)$ via exponential sums methods. It is well known that there are many
%close links between Fourier coefficients and various complexity
%characteristics of any Boolean function.

For a given input $X=(x_1,x_2,\dots,x_n)$, the sensitivity
$\sigma_{s}f$ on input $X$ is defined to be the number of bits such
that flipping one of them will change the value of the function.
Explicitly,
$$\sigma_{s,X}(f)=\sum_{i=1}^n\left |
 f(X)-f(X^{(i)})\right |,$$
 where $X^{(i)}=(x_1,\dots,x_{i-1},1-x_i,x_{i+1}\dots,x_n)$ denotes
the vector obtained from $X$ by flipping the i-th coordinate.
 The sensitivity $\sigma_{s}f$ of $f(X)$ is the maximum of $\sigma_{s,X}(f)$ on
  input $X$ over $\mathbb{Z}_2^n$.
 The average sensitivity
$\sigma_{av}(f)$ is defined to be the sensitivity average on all
inputs, i.e.,
$$\sigma_{av}(f)=2^{-n}\sum_{X\in
\mathbb{Z}_2^n}\sum_{i=1}^n\left |
 f(X)-f(X^{(i)})\right |.$$
   Sensitivity, and more generally, block sensitivity are important measures
    of complexity of Boolean functions.
%   It also has many applications in diverse fields especially in complexity theory.
It recently draws some attention. For instance,  Rubinstein and
Bernasconi showed large gaps between the average sensitivity and the
average block sensitivity \cite{R,Be}, Bernasconi, Damm and
Shparlinski gave the average sensitivity of testing square-free
numbers \cite{BD}, Boppana considered the average sensitivity of
bounded-depth circuits \cite{Bo} and Shi gave that the average
sensitivity as a lower bound of approximation polynomial degree, and
thus can serve as a lower bound of quantum query complexity
\cite{Shi}. For more details we refer to \cite{Bu}.

   In \cite{Sh} Shparlinski  studied the average sensitivity of laced Boolean function
   and gave a lower bound by obtaining a nontrivial bound on the Fourier coefficients of
the laced boolean function via exponential sums methods.
 He also gave two conjectures on the bounds of Fourier coefficients and the
average sensitivity of laced Boolean functions respectively.
Explicitly he conjectured that the average sensitivity of laced
Boolean functions on $n$ variables should be at least $(\frac
12+o(1))n$. He proved in the same paper that this value is greater
than $\gamma n$, where $\gamma\approx 0.135$ is a constant.

Recently an explicit formulas for the average sensitivity of laced
Boolean functions was given by D. Canright, etc . in \cite{CG} for
the case $p=n$ by using some counting formulas of the subset sums
over prime fields given by Li and Wan \cite{LW1}. Equivalently they
proved Shparlinski's conjecture for the case $p=n$. They also showed
further experimental evidence for the above conclusion on the
average sensitivity.

As we have known by the Prime Number Theorem, $p=n+o(n)$, see
\cite{B} for the best known result about gaps between prime numbers.
Thus, the gap between the present result and the expected result is
quite large. In this paper, we completely settled this problem. In
fact, we proved that the average sensitivity of laced Boolean
functions is $(\frac 12+o(1))n$.

We also compute the weights of the laced Boolean functions. Some
explicit formulas were given by D. Canright, S. Gangopadhyay, S.
Maitra and P. Stanica in \cite{CG} when $p-n\leq 3$ by the using the
same counting formulas of the subset sums over prime fields given in
\cite{LW1}. In this paper we extend their result for all the general
case.  We show that for a laced Boolean function $f$, the weight of
$f$ should be $2^{n-1}(1+o(1))$ and hence $f$ is asymptotically
balanced function.

This paper is organized as follows. In Section 2 we present a sieve
formula and we prove the main results in Section 3.

\section{A distinct coordinate sieving formula}

Our method is to evaluate a special exponential sum via a new
approach. The starting point is a new sieving formula discovered in
\cite{LW2}, which significantly improves the classical
inclusion-exclusion sieve in many interesting cases. We cite it here
without proof. For the details and some related applications we
refer to \cite{LW2,LW3}.

Let $D$ be a finite set, and let $D^k=D\times D \times \cdots \times
D$ be the Cartesian product of $k$ copies of $D$. Let $X$ be a
subset of $D^k$. %In many situations we are interested in counting
%the number of elements in the set
Define
\begin{equation}\label{1.0}\overline{X}=\{(x_1,x_2,\cdots,x_k)\in X
\ | \ x_i\ne x_j, \forall i\ne j\}.
\end{equation}
Let $f(x_1,x_2,\dots,x_k)$ be a complex valued function
defined over $X$.  Many problems arising from coding theory,
additive number theory and number theory are reduced to evaluate the
summation
\begin{eqnarray}
F=\sum_{x \in \overline{X}}f(x_1,x_2,\dots,x_k).\label{1.00}
\end{eqnarray}
Note that if we let $f(x_1,x_2,\dots,x_k)\equiv1$, then $F$ is just
the number of elements in $\overline{X}$.

 Let $S_k$ be the symmetric group on $\{1,2,\cdots, k\}$.
Each permutation $\tau\in S_k$ factorizes uniquely (up to the
  order of the factors) as a product of disjoint cycles and  each
  fixed point is viewed as a trivial cycle of length $1$.
 % For simplicity of the notation, we usually omit the $1$-cycles.
   Two permutations in $S_k$ are conjugate if and only if they
have the same type of cycle structure (up to the order).
 Let $C_k$ be the set of conjugacy  classes
 of $S_k$ and note that $|C_k|=p(k)$, the partition function. For a given $\tau\in S_k$,  let $l(\tau)$ be the number of
 cycles of $\tau$ including the trivial
  cycles. Then we can define the sign of $\tau$ to
  $\sign(\tau)=(-1)^{k-l(\tau)}.$
 For a given permutation $\tau=(i_1i_2\cdots i_{a_1})
  (j_1j_2\cdots j_{a_2})\cdots(l_1l_2\cdots l_{a_s})$
  with $1\leq a_i, 1 \leq i\leq s$, define
  \hskip 1.0cm
  \begin{eqnarray} \label{1.1}
     \ \ \ \ \ \ X_{\tau}=\left\{
(x_1,\dots,x_k)\in X,
 x_{i_1}=\cdots=x_{i_{a_1}},\cdots, x_{l_1}=\cdots=x_{l_{a_s}}
 \right\}.
   \end{eqnarray}
 Each element of $X_{\tau}$ is said to be of type $\tau$.
Thus $X_{\tau}$ is the set of all elements in $X$ of type $\tau$.
Similarly, for $\tau \in S_k$, we define
$$F_{\tau}=\sum_{x \in X_{\tau} }
f(x_1,x_2,\dots,x_k). $$

Now we can state our sieve formula. After that we will give one
corollary for the use of our proof. We remark that there are many
other interesting corollaries of this formula. For interested reader
we refer to \cite{LW2}.

\begin{thm} \label{thm1.0}We have
$$\label{1.5} F=\sum_{\tau\in
S_k}{\sign(\tau)F_{\tau}}.$$
 \end{thm}

%
% \begin{thm}\label{thm2.1}Let $\overline{X}, X_{\tau}$ be defined as
% in (\ref{1.0}) and (\ref{1.1}). Then we have
% \begin{eqnarray}\label{1.3}
%|\overline{X}| =\sum_{\tau\in S_k}{sign(\tau)|X_{\tau}}|.
%\end{eqnarray}
%  \end{thm}

Note that the symmetric group $S_k$ acts on $D^k$ naturally by
permuting coordinates. That is, for given $\tau\in S_k$ and
$x=(x_1,x_2,\dots,x_k)\in D^k$, we have
$$\tau\circ x=(x_{\tau(1)},x_{\tau(2)},\dots,x_{\tau(k)}).$$

 Before stating a useful corollary, we  first give two definitions.

\begin{defn} A subset $X$ in $D^k$ is said to be symmetric if for any $x\in X$ and
any $\tau\in S_k$, $\tau\circ x \in X $. %Furthermore, if $X$ satisfies the
%``strongly symmetric" condition, that is, for any $\tau$ and $\tau'$
%in $S_k$, one has $|X_{\tau}|=|X_{\tau'}|$ provided
%$l(\tau)=l(\tau')$, then we call $X$ a strongly symmetric set.
\end{defn}

\begin{defn}Let $X\subseteq D^k$ and assume $X$ is symmetric.
A complex-valued function $f(x_1,x_2,\dots, x_k)$ defined over $X$
is called normal on $X$ if for any two $S_k$-conjugate elements
$\tau$ and $\tau'$ in $S_k$ (thus $\tau$ and $\tau'$ have the same
type), we have
$$\sum_{x\in X_{\tau}
} f(x_1,x_2,\dots,x_k)=\sum_{x\in X_{\tau'} }
f(x_1,x_2,\dots,x_k).$$

%
%The function $f$ on $X$ is called strongly normal on $X$ if $X$ is
%strongly symmetric and for each $\tau$ and $\tau'$ in $S_k$, we
%always have
% $$\sum_{x\in X_{\tau}
%} f(x_1,x_2,\dots,x_k)=\sum_{x\in X_{\tau'} } f(x_1,x_2,\dots,x_k)$$
%whenever $l(\tau)=l(\tau')$.
\end{defn}

\textbf{Remark:}
%If $X$ is symmetric, that is, for any $\tau\in
%S_k$, $(x_1,x_2,\cdots,x_k)\in X$ if and only if
%$(x_{\tau(1)},x_{\tau(2)},\cdots,x_{\tau(k)})\in X$, then $X$ is
%normal, but not vice versa.
 If $f(x_1,x_2,\dots,x_k)$ is a symmetric function and $X$ is symmetric,
%$f(x_1,x_2,\cdots,x_k)=f(x_{\tau(1)},x_{\tau(2)},\cdots,x_{\tau(k)})$
%for any $\tau\in S_k$,
 then $f(x_1,x_2,\dots,x_k)$ must be normal
on $X$.
% Now take $f(x_1,x_2,\cdots,x_k)\equiv 1$, and note that $|X_{ij}|=\sum_{(x_1,x_2,\cdots,x_k)\in
%X_{ij} } 1 $. By applying the above theorem we obtain

%The main result of this paper is the following theorem, which gives
%an improvement of the enumeration formula of the inclusion-exclusion
%principle.  The number of terms in (\ref{1.3}) is reduced from
%$2^{{k \choose 2}}$ to $p(k)$, the number of partitions of $k$.
\begin{prop} \label{thm1.1}
 Let $C_k$ be the set of conjugacy classes
 of $S_k$.
 If $f$ is normal on $X$, then we have
$$F=\sum_{\tau \in C_k}(-1)^{k-l(\tau)}C(\tau)F_{\tau},$$
where $C(\tau)$ is the number of permutations conjugate to $\tau$.

%
%If $f$ is strongly normal on $X$,  for given $1 \leq i\leq k$, we
%choose $\tau_i\in S_k$ satisfying $l(\tau_i)=i$, and let
%$$F_{i}=\sum_{x\in X_{\tau_i}
%} f(x_1,x_2,\cdots,x_k)$$ (this is independent of the choice of
%$\tau_i$),  then we have
% \begin{eqnarray}
%F=\sum_{i=1}^{k}(-1)^{k-i}c(k,i)F_{i},
% \end{eqnarray}
%where $c(k,i)$ is the signless Stirling number of the first kind,
%that is,  the number of permutations in $S_k$ with exactly $i$
%cycles.
\end{prop}

For the purpose of our proof, we will also need the following
combinatorial formula. For simplicity we omit the proof and it can
be found in \cite{LW2}.

\begin{lem} \label{lem2.6}
Let $N(c_1,c_2,\dots,c_k)$ be the number  of permutations in $S_k$
of type $(c_1,c_2,\dots,c_k)$, that is,
$$N(c_1,c_2,\dots,c_k)=\frac
{k!} {1^{c_1}c_1! 2^{c_2}c_2!\cdots k^{c_k}c_k!},$$
 and define the generating function
\begin{align*}C_k(t_1,t_2,\dots,t_k)= \sum_{\sum
ic_i=k} N(c_1,c_2,\dots,c_k)t_1^{c_1}t_2^{c_2}\cdots t_k^{c_k}.
 \end{align*}
If $t_1=t_2=\cdots=t_k=q$, then we have
\begin{align} \label{6.2}
C_k(q,q,\dots,q) &=\sum_{\sum
ic_i=k} N(c_1,c_2,\dots,c_k)q^{c_1}q^{c_2}\cdots q^{c_k}\nonumber \\
&=(q+k-1)_k
 \end{align}
\end{lem}

\section{Subset sums on smooth subsets}
Let $D\subseteq \mathbb{Z}_p$ be a nonempty subset of cardinality
$n$. An additive character $\chi: Z_p\rightarrow \mathbb{C}^*$
 is a homomorphism from $Z_p$ to the non-zero complex numbers $\mathbb{C}^*$.  We define the Fourier bias of $D$ to be
$$\Phi(D)=\max_{\chi\ne \chi_0}\left |\sum_{a\in D}\chi{(a)}\right |.$$
Suppose that we have known $\Phi(D)$.
Let $N$ be the number of solutions of the equation
$$x_1+x_2+\cdots+x_k=b, x_i\in D, x_i\neq x_j, i\ne j.$$
In the following theorem we will give an asymptotic bound on $N$
when $\Phi({D})$ is small compared to $n=|D|$.
 \begin{thm}Let $N$ be the number of solutions of the equation
$$x_1+x_2+\cdots+x_k=b, x_i\in D, x_i\neq x_j, i\ne j.$$
Then we have $$\frac {N} {k!}\geq \frac 1p {n\choose k}-{\Phi(D)+k-1
\choose k}.$$
\end{thm}
\begin{proof}Let $X=D^k=D\times D \times \cdots \times D$ be the
Cartesian product of $k$ copies of $D$.
 Let $  \overline{X} =\left\{ (x_1,x_2,\dots,x_{k} )\in D^k \mid
 x_i\not=x_j,~ \forall i\ne j\} \right\}.$ It is clear that $|X|=n^k$ and
$|\overline{X}|=(n)_k$. Note that we have defined that for a
permutation $\tau\in S_k$, $X_{\tau}$ consists of the elements $x\in
X$ of type $\tau$.

Let  $G$ be the group of additive characters of $\mathbb{Z}_p$ and
$\chi_0$ be the trivial character. Then we deduce that
\begin{align*}
N&=\frac 1 {p} \sum_{(x_1, x_2,\dots x_k) \in \overline{X}}
\sum_{\chi\in G}\chi(x_1+x_2+\cdots +x_k-b)\\
&=\frac 1 {p}\sum_{\chi\in G} \sum_{(x_1, x_2,\dots x_k) \in
\overline{X}}\chi(x_1+x_2+\cdots +x_k-b)\\
&=\frac {(n)_k}{p}+\frac 1 {p} \sum_{\chi\ne \chi_0}\sum_{(x_1,
x_2,\cdots x_k) \in\overline{X}}\chi(x_1)\chi(x_2)\cdots \chi(x_k)\chi^{-1}(b)\\
&=\frac {(n)_k}{p}+\frac 1 {p} \sum_{\chi\ne
\chi_0}\chi^{-1}(b)\sum_{(x_1,x_2,\dots x_k) \in\overline{X}}\prod_{i=1}^{k} \chi(x_i).\\
\end{align*}
For given $\chi\ne \chi_0$, let $f_{\chi}(x)=
f_{\chi}(x_1,x_2,\dots,x_{k})= \prod_{i=1}^{k}\chi(x_i),$ and for
given $\tau$ let
$$F_{\tau}(\chi)=\sum_{x\in X_{\tau}}f_{\chi}(x)=\sum_{x \in X_{\tau}}\prod_{i=1}^{k} \chi(x_i).$$

Obviously $X$ is symmetric. It is also easy to check that
$f_{\chi}(x_1,x_2,\dots,x_{k})$ is normal on $X$. Thus by applying
Corollary \ref{thm1.1}, we deduce
  \begin{align*}
 N&=\frac {(n)_k}{p}+\frac 1 {p} \sum_{\chi\ne \chi_0}\chi^{-1}(b) \sum_{\tau\in
C_{k}}sign(\tau)C(\tau) F_{\tau}(\chi)\\
  \end{align*}
 where $C_{k}$ is the set of conjugacy classes
 of $S_{k}$, $C(\tau)$ is the number of permutations conjugate to $\tau$, and
 $F_{\tau}(\chi)=\sum_{x \in X_{\tau}}\prod_{i=1}^{k}\chi(x_i).$

For given $\tau\in C_{k}$, assume $\tau$ is of type
$(c_1,c_2,\dots,c_{k})$, where $c_i$ is the number of $i$-cycles in
$\tau$ for $1 \leq i\leq k$. Note that $\sum_{i=1}^{k} ic_i=k$ and
thus we deduce
  \begin{align*}
F_{\tau}(\chi)&=(\sum_{a\in D}\chi(a))^{c_1} (\sum_{a\in
D}\chi^2(a))^{c_2} \cdots
(\sum_{a\in D }\chi^{k}(a))^{c_{k}}\nonumber\\
&=\prod_{i=1}^{k}(\sum_{a\in D}\chi^i(a))^{c_i}\nonumber.\\
\end{align*}
By the definition of $\Phi(D)$ we have $F_{\tau}(\chi) \leq
(\Phi(D))^{\sum_{i=1}^{k}c_i}$ and thus
 \begin{align*}
 N&\geq\frac {(n)_k}{p}-\frac 1 {p} \sum_{\chi\ne \chi_0}
 \sum_{\tau\in C_{k}}C(\tau)(\Phi(D))^{\sum_{i=1}^{k}c_i}\\
 &=\frac {(n)_k}{p}-\frac {p-1} {p}\sum_{\sum ic_i=k} \frac
{k!} {1^{c_1}c_1! 2^{c_2}c_2!\cdots k^{c_{k}}c_{k}!}
(\Phi(D))^{\sum_{i=1}^{k}c_i}\\
&=\frac {(n)_k}{p}-\frac {p-1} {p}(\Phi(D)+k-1)_k. \qedhere
\end{align*}
The last equality is from Lemma \ref{lem2.6} and the proof is
complete.
\end{proof}

  If $n=p-c$, where $c$ is a fixed constant, then one checks that $\Phi(D)\leq c$
   and thus we get a clean and better bound, which was first found by \cite{LW1} by using
  elementary counting method.

 \begin{cor}
   If $n=p-c$, where $c$ is a fixed constant, then we have
 $$\frac {N} {k!}\geq \frac 1p {p-c \choose k}-{c+k-1 \choose k}.$$
 \end{cor}

Similarly we have
 \begin{cor}
 If $n=p-o(p)$, then $\Phi(D)\leq o(p)$ and thus we have
 $$\frac {N} {k!}\geq \frac 1p {p-o(p) \choose k}-{o(p)+k-1
\choose k}.$$
 \end{cor}

  \begin{cor}\label{cor3.4}
 Let $n=p-o(p)$. Let $N(b, D)$ be the number of subsets in $D$
 which sums to $b$. Then we have
 $$N(b, D)=\frac {2^n} p(1+o(1)).$$
 \end{cor}

We say a subset $D\subseteq A$ is smooth if $\Phi(D)=O(\sqrt{n\log
|A|})$, that is, for every nontrivial additive character $\chi$,
$|\sum_{a\in D}\chi{(a)}|=O(\sqrt{n\log |A|})$.
    \begin{cor}
    Let $D\subseteq\mathbb{Z}_p$ and $\epsilon$ be a positive
 constant.  If $|D|=\log^{1+\epsilon} p$ and $D$ is smooth, then there are
 two constants $c_1$ and $c_2$ such that if
$ c_1\frac {\log p} {\log\log{p}}\leq k \leq c_2 n $, then each
element $\mathbb{Z}_p$ can be written to be a k-subset sum in $D$.
    \end{cor}

    \begin{proof}
    The proof is left to the reader.
    \end{proof}

%\section{Fourier coefficients}
%The Fourier coefficients is a Boolean function is defined as
%$$\hat{f(u)}=\frac 1 {2^n} \sum_{x\in \mathbb{Z}_2^n}(-1)^{f(x)+u\dot x},$$
%where $u\in \mathbb{Z}_2^n$ and $u\cdot x$ is the usual inner
%product.

\section{Average sensitivity}
The average sensitivity $\sigma_{av}(f)$ of an n-variate Boolean
function $f(x_1,x_2,\dots,x_n)$ is defined as
$$\sigma_{av}(f)=2^{-n}\sum_{X\in \mathbb{Z}_2^n}\sum_{i=1}^n\left |
 f(X)-f(X^{(i)})\right |,$$
where $X^{(i)}$ is the vector obtained from $X$ by flipping its
$i$th coordinate. In \cite{Sh} the author asked the following
question.

For given $X=(x_1,x_2,\dots,x_n)\in \mathbb{Z}_2^n$, we define
$s(X)$ by the least positive integer of $\sum_{k=1}^n kx_k$ mod $p$,
i.e.,
$$s(X)\equiv \sum_{k=1}^n kx_k(\!\!\!\!\!\!\mod p),   1\leq s(X) \leq p.$$

Following \cite{SZ} we define the so called laced Boolean function
$$f(X)=\left\{
\begin{array}{ll}
  x_{s(X)}, \ \  1\leq s(X)\leq n;\\
  x_1, \ \  \hbox{otherwise.}\\
    \end{array}
    \right.
    $$

We first compute the weight of $f(X)$. We have the following
theorem, which significantly improve the results given by \cite{CG}.
\begin{thm}Let $f(X)$ be defined as above. Then we have
$$wt(f)={2^{n-1}}(1+o(1)).$$
In other words, $f(X)$ is an asymptotically balanced function..
\end{thm}

\begin{proof}
Let $A=\{0,n+1,n+2,\dots,p-1\}$ and $D=Z_p\setminus A$. By applying
Corollary \ref{cor3.4} we have
 \begin{align*}
 wt(f)=\sum_{X\in Z_2^n}f(X)&=\sum_{s=1}^n\sum_{X\in Z_2^n, s(X)=s, x_s=1}1+
\sum_{s=n+1}^p\sum_{X\in Z_2^n, s(X)=s, x_1=1}1\\
&=\sum_{s=1}^n N(0,D\backslash \{s\})+\sum_{s=n+1}^pN(s-1,D\backslash \{1\})\\
&=\sum_{s=1}^p \frac 1 p  2^{n-1}(1+o(1))\\
 &={2^{n-1}}(1+o(1)) \qedhere
 \end{align*}
\end{proof}

In \cite{Sh} Shparlinski studied $\sigma_{av}(f)$ and raised the
following conjecture:
\begin{conj}
Is it true that for the function given by (1) we have
$$\sigma_{av}(f)\geq \left( \frac 12 +o(1)\right)n?$$
\end{conj}

In the same paper Shparlinski gave a lower bound by obtaining a
nontrivial bound on the Fourier coefficients of the laced boolean
function via exponential sums methods. He proved in the same paper
that this value is greater than $\gamma n$, where $\gamma\approx
0.135$ is a constant.

Recently an explicit formulas for the average sensitivity of laced
Boolean functions was given by D. Canright, etc . in \cite{CG} for
the case $p=n$ by using some counting formulas of the subset sums
over prime fields given in \cite{LW1}. Equivalently they proved
Shparlinski's conjecture for $p=n$. They also showed further
experimental evidence for the above conclusion on the average
sensitivity.

%As we have known by the Prime Number Theorem, $p=n+o(n)$, see
%\cite{B} for the best known result about gaps between prime numbers.
%Thus, the gap between the present result and the expected result is
%quite large. In this paper, we completely settled this problem. We
%proved that the average sensitivity of laced Boolean functions is
%$(\frac 12+o(1))n$.

We will prove this conjecture now. In fact we obtain a stronger
result.
\begin{thm}Let $\sigma_{av}(f)$ be the average sensitivity of the
laced Boolean function. Then
$$\sigma_{av}(f)=\left( \frac 12 +o(1)\right)n.$$
\end{thm}
\begin{proof}Let $A=\{0,n+1,n+2,\dots,p-1\}$ and $D=Z_p\setminus A$.
Since we have the symmetry between the bits 1 and 0,
 for simplicity we just need to consider the number of bit changes
from 0 to 1. Thus we have
 \begin{align*}
2^{n-1}\sigma_{av}(f)&=\sum_{X\in Z_2^n}\sum_{i=1}^n\left |
 f(X)-f(X^{(i)})\right |\\
 &=\sum_{s\in D}\sum_{X\in Z_2^n, s(X)=s, x_s=1}\sum_{i=1}^n\left |
    1-f(X^{(i)})\right |
    +\sum_{s\in D}\sum_{X\in Z_2^n, s(X)=s, x_s=0}\sum_{i=1}^n\left |
 0-f(X^{(i)})\right |\\
&+\sum_{s\in A}\sum_{X\in Z_2^n, s(X)=s, x_1=1}\sum_{i=1}^n\left |
 1-f(X^{(i)})\right |+\sum_{s\in A}\sum_{X\in Z_2^n, s(X)=s, x_1=0}\sum_{i=1}^n\left |
 0-f(X^{(i)})\right |\\
 &=\sum_{i=1}^n\sum_{s=1}^n  \sum_{X\in Z_2^n, s(X)=s, x_i=0, x_s=1,
 x_{s+i}=0}1+\sum_{i=1}^n\sum_{s=1}^n  \sum_{X\in Z_2^n, s(X)=s, x_i=0, x_s=1,
 x_{s+i}=1}1\\
 &+\sum_{i=1}^n\sum_{s=n+1}^p \sum_{X\in Z_2^n, s(X)=s, x_i=0, x_1=1,
 x_{s+i}=0}1+\sum_{i=1}^n\sum_{s=n+1}^p  \sum_{X\in Z_2^n, s(X)=s, x_1=0, x_s=1,
 x_{s+i}=1}1\\
 \end{align*}
Recall that from the Prime Number Theorem we have $p=n+o(n)$.  We
notice that for simplicity we may assume that in the 4 summations,
$s+i(\!\!\!\mod p)\in D$, otherwise the summations is bounded by
$o(1)$ respectively (for instance: we only need to replace $x_{s+i}$
 to $x_1$) and thus can be omitted.
Thus we have
 \begin{align*}
 2^{n-1}\sigma_{av}(f)
 &=\sum_{i=1}^n\sum_{s=1}^nN(0,D\backslash \{i, s+i, s \})+
 \sum_{i=1}^n\sum_{s=1}^nN(0,D\backslash \{i, s-i, s \})\\
 &+ \sum_{i=1}^n\sum_{s=n+1}^pN(0,D\backslash \{i, s+i, 1 \})+
 \sum_{i=1}^n\sum_{s=n+1}^pN(0,D\backslash \{i, s-i, 1 \})\\
 &\approx\sum_{i=1}^n\sum_{s=1}^n \frac {2^{n-2}} p(1+o(1))+
 \sum_{i=1}^n\sum_{s=n+1}^p \frac {2^{n-2}} p(1+o(1))\\
 &\approx n{2^{n-2}} (1+o(1))+o(n){2^{n-2}} (1+o(1))
 %+\sum_{s=1}^n  \sum_{X\in Z_2^n, s(X)=s, x_s=0}\sum_{i=1}^n\left |
% (x_i+x_{s+i})\right |\\
% \sum_{i=1}^n\sum_{s=n+1}^p\sum_{X\in Z_2^n, s(X)=s}\left |
% x_1-f(X^{(i)})\right |\\
% =\sum_{i=1}^n\sum_{s=1}^n\left ( \sum_{X\in Z_2^n, s(X)=s, x_i=0}\left |
% x_s-f(X^{(i)})\right |+ ( \sum_{X\in Z_2^n, s(X)=s, x_i=1}\left |
% x_s-f(X^{(i)})\right | \right )+\sum_{i=1}^n\sum_{s=n+1}^p\sum_{X\in Z_2^n, s(X)=s}\left |
% x_1-f(X^{(i)})\right |\\
\end{align*}
Thus we have $$\sigma_{av}(f)=\left(\frac 12
+o(1)\right)n.\qedhere$$
%We note that we only compute the cases that $f(X)$ flips from 1 to 0
%since the cases from 0 to 1 are similar.
%\begin{eqnarray*}
% &=&\sum_{i=1}^n\sum_{s\in D}\sum_{X\in Z_2^n,s(X)=s,x_s=1}\left |
% f(X)-f(X^{(i)})\right|\\
% &=& \sum_{i=1}^n\sum_{s\in D}\sum_{X\in Z_2^n,s(X)=s,x_s=1,x_i=0}(
% 1-x_{b+i \mod p})+\sum_{i=1}^n\sum_{s\in D}\sum_{X\in Z_2^n,s(X)=s,x_s=1,x_i=1}(
% 1-x_{b-i \mod p})\\
%\end{eqnarray*}
\end{proof}


\begin{thebibliography}{00}

%\bibitem{a1}  T. Cormen, C. Leiserson, R. Rivest,
% \emph{Introduction to Algorithms.}, MIT
%Press, 2001. Chapter 35.5, The subset-sum problem.

%\bibitem{C} Q. Cheng, \emph{Hard problems of algebraic geometry codes}, IEEE Trans. \& Inform Theory, 2008, 54(1), 402-406.
%
%\bibitem{CM} Q. Cheng and E. Murray, \emph{On deciding deep holes of Reed-Solomon codes},
% In: TAMS 2007, Lecture Notes in Computer Science, Vol.4484,  Springer, 2007.

%\bibitem{CW1} Q. Cheng and D. Wan, \emph{On the list and bounded distance decodability
%of Reed-Solomon codes}, FOCS (2004), 335-341.

%\bibitem{CW2} Q. Cheng and D. Wan,
 % \emph{On the list
%and bounded distance decodability of Reed-Solomon codes},
% SIAM J. Comput. \textbf{37} (2007), no. 1, 195-209.

\bibitem{B} R.C. Baker, G. Harman and J. Pintz, \emph{The difference between
consecutive primes, II}, Proc. London Math. Soc., 83 (2001)
532--562.

\bibitem{Be} A. Bernasconi, \emph{Sensitivity vs. block sensitivity (an average-case
study)}, Information Processing Letters, 59 (1996), 151--157.

\bibitem{BD} A. Bernasconi, C. Damm and I. Shparlinski, \emph{
The average sensitivity of square-freeness}, Comput. Complexity 9
(2000), 39--51.

\bibitem{Bo} R. B. Boppana, \emph{The average sensitivity of bounded-depth
circuits}, Information Processing Letters, 63 (1997), 257--261.

\bibitem{Bu} H. Buhrman  and R. de Wolf, \emph{Complexity
measures and decision tree complexity: a survey}, Complexity and
logic (Vienna, 1998), Theoret. Comput. Sci. 288 (2002), 21--43.

\bibitem{CG} D. Canright, S. Gangopadhyay, S. Maitra and P. Stanica,
\emph{Laced Boolean functions and subset sum problems in finite
fields}, Discrete Applied Mathematics, 159 (2011), 1059--1069.

%\bibitem{G} R. L. Graham, D. E. Knuth and O. Patashnik, \emph{Concrete Mathematics: A Foundation for
%Computer Science}, 2nd ed. Reading, MA: Addison-Wesley, 1994.

\bibitem{LW1} J. Li and D. Wan, \emph{On the subset sum problem over finite fields},
Finite Fields \& Applications, {14} (2008), 911--929.

\bibitem{LW2} J. Li and D. Wan, \emph{A new sieve for distinct coordinate counting},
Science in China Series A, 53 (2010), 2351--2362.

\bibitem{LW3} J. Li and D. Wan, \emph{Counting subsets of finite Ablelian groups},
to appear in JCTA.

%\bibitem{N} C. A. Nicol , \emph{Linear congruences and the Von
%Stemeck function}, Duke Math. J. 26 (1959), 193--197.

%\bibitem{O} A.M. Odlyzko and R.P. Stanley,\emph{Enumeration of power sums
%modulo a prime}, J. Number Theory 10 (1978), 263--272.

%\bibitem{R}   K. G. Ramanathan,  \emph{Some applications of Ramanujan's trigonometrical sum
%$C_m(n)$}, Proceedings of the Indian Academy of Sciences, vol. 20
%(1945), 62--69.

\bibitem{R} D. Rubinstein, \emph{Sensitivity vs. block sensitivity of Boolean
functions}, Combinatorica 15 (1995), 297--299.

\bibitem{SZ} P. Savicky and S. Zak, \emph{A read-once lower bound and a $(1,+k)$-hierarchy
for branching programs}, Theoret. Comput. Sci. 238 (2000), 347--362.

\bibitem{SD} M. Sauerhoff and D. Sieling, \emph{Quantum branching programs and space-bounded nonuniform quantum
complexity}, Theoret. Comput. Sci., 334 (2005), 177--225.

\bibitem{S} M. Sauerhoff, \emph{Randomness versus nondeterminism for read-once and read-k branching
programs}, STACS 2003, 307--318, Lecture Notes in Comput. Sci.,
2607, Springer, Berlin, 2003.

\bibitem{Shi} Y. Shi, \emph{Lower bounds of quantum black-box complexity and degree of
approximating polynomials by influence of Boolean variables},
Information Processing Letters, 75 (2000), 79--83.

\bibitem{Sh} Igor. E. Shparlinski, \emph{Bounds on the Fourier coefficients
of the weighted sum function}, Inform. Process. Lett. 103 (2007),
83--87.

%\bibitem{St} R.P. Stanley, \emph{Enumerative combinatorics. Vol. 1},
%Cambridge University Press, Cambridge, 1997.

%\bibitem{N}Alon,N.(1999) \emph{Combinatorial Nullstellensatz}, Combinatorics, Probability and
%Computing,8,7-29
%\bibitem{S} Wolfgang M.Schmidt, \emph{Equations over
%Finite Fields: An Elementeray Approach}, Springer-Verlag, 1976.

%\bibitem{Wan} D. Wan, \emph{Generators and irreducible polynomials over finite
%fields}, Math. Comp. \textbf{66} (1997), no. 219, 1195--1212.

%\bibitem{W} Weil, Andr\'{e}, \emph{ Numbers of solutions of equations in finite
%fields}, Bull. Amer. Math. Soc. \textbf{55}, (1949).
\end{thebibliography}
\end{document}